\documentclass[11pt,a4paper]{article}

\usepackage[utf8]{inputenc}
\usepackage[english]{babel}
\usepackage{amsmath}
\usepackage{amsfonts}
\usepackage{amsthm}
\usepackage{amssymb}
\usepackage{graphicx}
\usepackage{mathtools}
\usepackage{parskip}
\usepackage[linkcolor=black,colorlinks=true,citecolor=black,filecolor=black]{hyperref}
\usepackage{microtype}
\usepackage{xcolor}
\usepackage{url}

\newtheorem{theorem}{Theorem}

\newtheorem{definition}[theorem]{Definition}
\newtheorem{lemma}[theorem]{Lemma}

\newtheorem{corollary}[theorem]{Corollary}

\newcommand{\Z}{\mathbb{Z}}
\renewcommand{\epsilon}{\varepsilon}

\newcommand{\DeclareAutoPairedDelimiter}[3]{%
  \expandafter\DeclarePairedDelimiter\csname Auto\string#1\endcsname{#2}{#3}%
  \begingroup\edef\x{\endgroup
    \noexpand\DeclareRobustCommand{\noexpand#1}{%
      \expandafter\noexpand\csname Auto\string#1\endcsname*}}%
  \x}
\DeclareAutoPairedDelimiter\ceil{\lceil}{\rceil}
\DeclareAutoPairedDelimiter\floor{\lfloor}{\rfloor}

\newcommand{\poly}{\operatorname{poly}}

\newcommand{\OPT}{\operatorname{OPT}}

\title{Quantum Space--Time Tradeoffs for TSP and Hypercube Path via Extremal Set Systems}

\author{Justin Dallant\thanks{Email: \href{mailto:justin.dallant@tu-dresden.de}{justin.dallant@tu-dresden.de}}\\
Dresden University of Technology}

\date{}

\begin{document}

\maketitle

\begin{abstract}
We give a quantum analogue of the extremal set-system framework introduced by Ameli, Nederlof and Wang (FOCS 2026) and Dallant and Kozma (FOCS 2026) for space--time tradeoffs in TSP and related permutation problems. Let $1<S\leq 2$, and let $P_S$ denote the optimal inverse normalized chain density among set systems of normalized size at most $S$. We use quantum minimum finding over a covering family of set systems to obtain a bounded-error algorithm using $S^{n+o(n)}$ QRAM space and $\left(S\sqrt{P_S}\right)^{n+o(n)}$ time. Combined with the set-system constructions of Andoni, Dallant, Kozma and Yu, this gives an explicit tradeoff improving the Caroppo et al.\ (ESA 2026) TSP tradeoff for $1<S\le1.657$. The same method applies to a variety of permutation problems, including the Hypercube Path query problem introduced by Ambainis et al.\ (SODA 2019), for which it improves the previous bound of $1.817^{n+o(n)}$ for both time and QRAM space to $1.784^{n+o(n)}$ time using $1.430^{n+o(n)}$ QRAM space.
\end{abstract}

\medskip
{\paragraph*{Acknowledgements} \small I thank Jevgēnijs Vihrovs for pointing out that the results presented here also directly improve the quantum tradeoffs and worst-case running time for the Hypercube Path query problem.}

\section{Introduction}

The Traveling Salesman Problem (TSP) is one of the flagship problems in exact
exponential-time algorithms.  In its standard form, we are given a complete
weighted graph and asked to find a minimum-weight tour visiting every vertex
exactly once.  We work with the minimum-weight Hamiltonian-path variant with prescribed
start and end vertices.  The standard tour version reduces to this variant with
only polynomial overhead, so we continue to refer to the problem as TSP.

The classical dynamic-programming algorithm of Bellman, Held and Karp
\cite{Bellman1962,HeldKarp1962} solves TSP in $2^{n+O(\log n)}$ time and
space.  It computes, for every subset of cities and every possible endpoint in that subset, the length of the
shortest path that starts at a fixed vertex, visits precisely that subset, and
ends at the endpoint.  Despite being more than sixty years old, the
$O(2^n\poly(n))$ bound remains the best known for general weighted TSP.

The exponential space usage of Bellman--Held--Karp is often a more severe
bottleneck than its exponential running time.  At the other extreme, a
divide-and-conquer algorithm of Gurevich and Shelah
\cite{GurevichShelah1987} solves Hamiltonian path, and hence TSP by standard
modifications, in $4^{n+O(\log^2 n)}$ time using only polynomial space.\footnote{The runtime can be improved to $4^{n+O(\log n)}$ while keeping polynomial space by a folklore argument, explicitly described by Ameli, Nederlof and Wang in \cite[Appendix A]{AmeliNederlofWang2026}.}
Interpolating between Bellman--Held--Karp and Gurevich--Shelah yields a family
of tradeoff points with time $T^{n+o(n)}$, space $S^{n+o(n)}$, and
$ST=4$, for various values of $1\le S\le 2$.

Koivisto and Parviainen \cite{KoivistoParviainen2010} showed that the product
$ST=4$ is not a fundamental barrier: using a partition of permutations into
classes given by linear extensions of certain posets, they obtained tradeoff
points with $ST<4$, with minimum product about $3.93$.  More recently,
Ameli, Nederlof and Wang \cite{AmeliNederlofWang2026} and, independently,
Dallant and Kozma \cite{DallantKozma2026} introduced a more flexible framework
based on set systems containing many maximal chains.  In the notation used
below, their framework gives deterministic classical algorithms with space
$S^{n+o(n)}$ and time
\[
        (SP_S)^{n+o(n)},
\]
where $P_S$ measures the inverse normalized chain density of the best set
system of normalized size at most $S$.  Dallant and Kozma obtain tradeoff
points with $ST<3.572$, while Ameli, Nederlof and Wang concurrently obtained $ST<3.75$ via
an equivalent chain-efficiency formulation. An updated version of Ameli, Nederlof and Wang's paper obtains $ST<3.1861$. Andoni et al.\
improved this to $ST<3.182$ and showed that this is essentially optimal within
the classical set-system framework~\cite{AndoniDallantKozmaYu}.  

TSP has also received attention in the quantum setting.  Ambainis et
al.~\cite{AmbainisEtAl2019} combined quantum minimum finding with dynamic
programming to obtain a bounded-error TSP algorithm running in
$1.728^{n+o(n)}$ time and QRAM space.  They also introduced
the Hypercube Path query problem to capture more general permutation and
vertex-ordering problems, obtaining $1.817^{n+o(n)}$ time and space for this problem.
Caroppo et al.~\cite{CaroppoEtAl2026} subsequently studied quantum space--time
tradeoffs for these problems, motivated by the cost of QRAM\@.

Here we obtain a quantum speedup over the corresponding classical set-system approach by replacing the outer enumeration of the
set-system cover by quantum minimum finding.  The classical set-system-based algorithm evaluates
about $P_S^n$ restricted Bellman--Held--Karp dynamic programs, each of size
about $S^n$.  Quantum minimum finding reduces the covering factor from $P_S^n$
to $P_S^{n/2}$ while leaving the cost and space of each restricted dynamic
program unchanged up to polynomial factors.  This gives time
\[
        \left(S\sqrt{P_S}\right)^{n+o(n)}
\]
and space $S^{n+o(n)}$.

\subsection*{Main results}

For a problem on $n$ elements, call $(S,T)$ a feasible quantum space--time tradeoff
point if it can be solved with bounded error in time $T^{n+o(n)}$ using
$S^{n+o(n)}$ QRAM space.  We show the following (see Section~\ref{sec:perm-problems} for the definition of a min-type permutation problem).

\begin{theorem}
\label{thm:main}
For every $1<S\le 2$, TSP, the Hypercube Path query problem and, more
generally, bounded-degree min-type permutation problems admit bounded-error
quantum algorithms using $S^{n+o(n)}$ QRAM space and
\[
        \left(S\sqrt{P_S}\right)^{n+o(n)}
\]
time.
\end{theorem}

Combining Theorem~\ref{thm:main} with the explicit
set-system bounds of Andoni et al.~\cite{AndoniDallantKozmaYu} and quantum divide and conquer (see Section~\ref{sec:fractalization}) gives the
tradeoff curve stated in Corollary~\ref{cor:explicit}.  In particular, their
bound at $S=2^{0.515617}$ gives, for TSP, Hypercube Path, and all bounded-degree
min-type permutation problems, QRAM space
\[
        2^{0.515617n+o(n)}
        \le 1.430^{n+o(n)}
\]
and time
\[
        2^{(0.515617+0.638624/2)n+o(n)}
        \leq 1.784^{n+o(n)}.
\]
For Hypercube Path, this improves the previous $1.817^{n+o(n)}$
running-time bound of Ambainis et al.~\cite{AmbainisEtAl2019}, which uses
$1.817^{n+o(n)}$ QRAM space.

\section{Preliminaries}

\subsection{Permutation problems}
\label{sec:perm-problems}

We recall the standard definition of permutation problems, following Koivisto
and Parviainen.

\begin{definition}
Let $R$ be a semiring with addition $\oplus$ and multiplication
$\otimes$.  A permutation problem of degree $d$ over $R$ is specified by
local functions $f_1,\ldots,f_n$, and the task is to compute
\[
        \bigoplus_{\pi\in S_n} f(\pi),
\]
where
\[
        f(\pi)
        =
        \bigotimes_{j=1}^n
        f_j\!\left(
            \{\pi_1,\ldots,\pi_j\},
            \pi_{\max\{1,j-d+1\}},\ldots,\pi_j
        \right).
\]
We assume throughout that $d=O(1)$, and that the local functions and semiring
operations are computable in polynomial time.
\end{definition}

The classical set-system framework applies particularly naturally to
permutation problems over additively idempotent semirings, since duplicate
coverage of a permutation does not affect the value of the outer
$\oplus$-sum.  The quantum speedup considered here, however, uses
quantum minimum finding at the outer level.  It therefore requires the outer aggregation to be expressible as a minimum.

\begin{definition}
A bounded-degree permutation problem is called \emph{min-type} if its value can
be written as
\[
        \operatorname{val}(I)=\min_{\pi\in S_n} f_I(\pi),
\]
where $f_I(\pi)$ has the constant-degree local decomposition above.
Equivalently, the addition operation $\oplus$ of the associated semiring is
$\min$.
\end{definition}

The version of TSP with fixed endpoints is a min-type permutation problem over
the $(\min,+)$ semiring.  Other optimization problems expressible through a minimum over vertex orderings, such as Feedback Arc Set and several
graph layout problems, also fit this template.

Ambainis et al.\ similarly introduced the Hypercube Path query problem as a way to capture a broad class of dynamic-programming problems, including vertex-ordering problems such as treewidth, pathwidth, and sum cut~\cite{AmbainisEtAl2019}. In this problem, the input is a subgraph of the $n$-dimensional Boolean hypercube accessible through edge queries, and the task is to determine whether it contains a monotone path from $0^n$ to $1^n$. Assign cost $0$ to a present hypercube edge and cost $1$ to an absent one. Under the standard correspondence between monotone hypercube paths and permutations, the Hypercube Path query problem is then a min-type permutation problem over the $(\min,+)$ semiring. Hence our results apply to it as well.

\subsection{Set systems framework}
\label{sec:set-systems}

A set system over $[m]$ is a family $\mathcal F\subseteq 2^{[m]}$.  We call any sequence
\[
        \emptyset = F_0 \subset F_1 \subset \cdots \subset F_m = [m],
        \qquad F_i\in \mathcal F,\quad |F_i|=i.
\]
a maximal chain of $\mathcal F$.
Let $C(\mathcal F)$ denote the number of maximal chains of $\mathcal F$.

Every permutation $\pi=(\pi_1,\ldots,\pi_m)$ of $[m]$ defines a maximal
chain of the full Boolean lattice by its prefix sets
\[
        \pi^{(i)} := \{\pi_1,\ldots,\pi_i\},
        \qquad i=0,\ldots,m.
\]

\begin{definition}
A permutation $\pi\in S_m$ is supported by a set system
$\mathcal F\subseteq 2^{[m]}$ if all its prefix sets lie in $\mathcal F$;
that is,
\[
        \pi^{(i)}\in\mathcal F
        \qquad\text{for all }i=0,\ldots,m.
\]
Equivalently, the maximal chain corresponding to $\pi$ is contained in
$\mathcal F$.
\end{definition}

Following Dallant and Kozma,\footnote{We use the terminology and notation of
Dallant and Kozma~\cite{DallantKozma2026}.  Most of the set-system results also appear in equivalent form in the paper by Ameli, Nederlof
and Wang~\cite{AmeliNederlofWang2026}.} we define the normalized size and
inverse normalized chain density of $\mathcal F$ by
\[
        S(\mathcal F)=|\mathcal F|^{1/m},
        \qquad
        P(\mathcal F)=\left(\frac{m!}{C(\mathcal F)}\right)^{1/m},
\]
with $P(\mathcal F)=+\infty$ if $C(\mathcal F)=0$.  For $1<S\le 2$, define
\[
        P_S
        =
        \inf\{P(\mathcal F): S(\mathcal F)\le S\},
\]
where the infimum ranges over all set systems $\mathcal F\subseteq 2^{[m]}$ and all $m$.

The quantity $S(\mathcal F)$ measures how much dynamic-programming space is
needed when restricting to $\mathcal F$.  The quantity $P(\mathcal F)$
measures how many relabelings of $\mathcal F$ are needed, up to subexponential
factors, to cover all permutations.

We use the following covering theorem.

\begin{theorem}[\cite{DallantKozma2026}]
\label{thm:covering}
For every $1<S\le 2$ and every sufficiently large $n$, there is a family
\[
        \mathcal F_1,\ldots,\mathcal F_q
        \subseteq 2^{[n]}
\]
such that
\[
        q \le (P_S+o(1))^n,
        \qquad
        |\mathcal F_j| \le S^{n+o(n)}
        \quad\text{for all }j,
\]
and every permutation of $[n]$ is supported by at least one member of the
family.  Moreover, the family can be generated deterministically on demand with only
subexponential overhead: given $j$, a member $\mathcal F_j$ can be constructed
within $S^{n+o(n)}$ time and space.
\end{theorem}

The following lemma follows from a straightforward modification of the
Bellman--Held--Karp dynamic-programming algorithm.

\begin{lemma}[\cite{DallantKozma2026}]
\label{lem:restricted-dp}
Let $\mathcal F\subseteq 2^{[n]}$.  Given a TSP instance on $n$ vertices,
the minimum-weight TSP path among permutations supported by $\mathcal F$ can be computed
in
\[
        O(|\mathcal F|\poly(n))
\]
time and space.  More generally, any bounded-degree min-type permutation problem
can be restricted to permutations supported by $\mathcal F$ with the same
$O(|\mathcal F|\poly(n))$ time and space bound.
\end{lemma}

\subsection{Quantum model}

We work in the standard quantum circuit model augmented with writable QRAM\@.  A QRAM of
size $N$ stores $N$ words of polynomially many bits or qubits and allows
coherent access in superposition.  We assume that a QRAM access costs
$\poly(n)$ time.

Our algorithm uses QRAM as workspace for the restricted dynamic program inside
the evaluation oracle for quantum minimum finding.  Since the oracle must be
invoked coherently on superpositions of indices $j$, the restricted dynamic
program is implemented reversibly, with QRAM used as quantum read-write
workspace.  

We use the following standard primitive of Dürr and Høyer
\cite{DurrHoyer1996}.

\begin{theorem}[\cite{DurrHoyer1996}]
\label{thm:qmf}
Given oracle access to a function $A:[N]\to  R$, where $R$ is a totally ordered set of efficiently comparable values, there is a
bounded-error quantum algorithm that finds $\min_{i\in[N]} A(i)$ using
$O(\sqrt N)$ oracle calls, up to polylogarithmic overhead.
\end{theorem}

\section{Quantum tradeoffs from extremal set systems}
\label{sec:quantum-tradeoff}

\begin{proof}[Proof of Theorem~\ref{thm:main}]
Fix $1<S\le 2$.  By Theorem~\ref{thm:covering}, for every sufficiently large
$n$ there is a family
\[
        \mathcal F_1,\ldots,\mathcal F_q
\]
of set systems over $[n]$ such that
\[
        q\le (P_S+o(1))^n,
        \qquad
        |\mathcal F_j|\le S^{n+o(n)}
\]
for every $j$, and every permutation of $[n]$ is supported by at least one
$\mathcal F_j$.

For a set system $\mathcal F_j$, let $A(j)$ be the value of the minimum-weight TSP
path whose corresponding permutation is supported by $\mathcal F_j$.  By
Lemma~\ref{lem:restricted-dp}, the value $A(j)$ can be computed
deterministically in time and space
\[
        S^{n+o(n)}.
\]
By standard reversible simulation, with compute--copy--uncompute, the on-demand construction of $\mathcal F_j$ and the restricted dynamic program can be implemented reversibly without changing the exponential base, giving a valid quantum evaluation oracle for $A(j)$.

Because every permutation is supported by at least one member of the covering
family, the global optimum is
\[
        \OPT
        =
        \min_{1\le j\le q} A(j).
\]
Indeed, every solution considered by some $A(j)$ is a genuine TSP solution, so
$\min_j A(j)\ge \OPT$.  Conversely, an optimal permutation is supported by at
least one $\mathcal F_j$, so $\min_j A(j)\le \OPT$.

We compute this minimum using quantum minimum finding over $[q]$.  The number
of oracle calls is
\[
        O\!\left(\sqrt q\right)
        =
        \left(\sqrt{P_S}\right)^{n+o(n)}.
\]
Each oracle call costs $S^{n+o(n)}$ time and uses
$S^{n+o(n)}$ workspace.  The total running time is therefore
\[
        S^{n+o(n)}\sqrt q
        =
        \left(S\sqrt{P_S}\right)^{n+o(n)},
\]
while the QRAM space remains $S^{n+o(n)}$, up to polynomial overhead.

The same proof applies to any bounded-degree min-type permutation problem: for
each $j$, define $A(j)$ to be the minimum value among solutions supported by
$\mathcal F_j$.  Since the covering family supports every permutation, the
global optimum is again $\min_j A(j)$, and quantum minimum finding applies.
\end{proof}

Note that the result applies more narrowly than the classical analogue. In the classical set-system
framework one can compute
\[
        \bigoplus_{j=1}^q A(j),
\]
and idempotence ensures that duplicate coverage of a permutation is harmless.
The quantum argument above instead replaces the outer loop by quantum minimum finding, which does not apply to a generic semiring addition $\oplus$.

\section{Fractalization}
\label{sec:fractalization}

The tradeoff of Theorem~\ref{thm:main} can be improved at small space by combining it with quantum divide and conquer, in a manner called quantum fractalization by Caroppo et al.~\cite{CaroppoEtAl2026}.
As bounded-degree min-type permutation problems satisfy the balanced-splitting condition required for their fractalization lemma, their result can be stated as follows in our setting.

\begin{theorem}[\cite{CaroppoEtAl2026}]
\label{thm:fractalization}
Suppose $(S,T)$ is a feasible quantum space--time tradeoff point for a bounded-degree min-type permutation problem.
Then
\[
        \left(\sqrt S,\sqrt{2T}\right)
\]
is also feasible.
\end{theorem}

Iterating this result gives the following.

\begin{corollary}
\label{cor:iterated-fractalization}
If $(S_0,T_0)$ is feasible for TSP or for a bounded-degree min-type permutation
problem, then for every integer $r\ge 0$,
\[
        \left(
            S_0^{2^{-r}},
            2^{1-2^{-r}} T_0^{2^{-r}}
        \right)
\]
is feasible.
\end{corollary}

Combining Corollary~\ref{cor:iterated-fractalization} with
Theorem~\ref{thm:main} yields the following general bound.

\begin{corollary}
\label{cor:fractalized-set-system}
For every $1<S\le 2$, TSP and every bounded-degree min-type permutation
problem admit bounded-error quantum algorithms using $S^{n+o(n)}$ QRAM
space and $T(S)^{n+o(n)}$ time, where
\[
        T(S)
        \le
        \min_{\substack{r\in\Z_{\geq 0}\\ S^{2^r}\le 2}}
        2^{1-2^{-r}}
        \left(
            S^{2^r}\sqrt{P_{S^{2^r}}}
        \right)^{2^{-r}} .
\]
\end{corollary}

\section{Explicit bounds from improved set systems}
\label{sec:adky-bound}

We now apply the explicit extremal set-system bound of Andoni et al.~\cite{AndoniDallantKozmaYu}.

\begin{theorem}[\cite{AndoniDallantKozmaYu}]
\label{thm:adky-interpolated}
Let $1<S\le 2$, let
\[
        k=\left\lfloor \frac{1}{\log_2 S}\right\rfloor,
        \qquad
        \lambda=k(k+1)\log_2 S-k.
\]
There is a constant $I\ge 0.654865$ such that
\[
        P_S
        \le
        k^{\lambda}(k+1)^{1-\lambda}
        2^{-I(1-\log_2 S)}.
\]
Moreover, for $S=2^{0.515617}$,
\[
        P_S \le 2^{0.638624}.
\]
\end{theorem}

For a fully explicit bound, fix $I_0=0.65486$.  Since
$I\ge 0.654865>I_0$, Theorem~\ref{thm:adky-interpolated} implies
\[
        P_S \le \widehat P(S)
        :=
        k^{\lambda}(k+1)^{1-\lambda}
        2^{-I_0(1-\log_2 S)},
\]
where $k=\left\lfloor \frac{1}{\log_2 S}\right\rfloor$ and $\lambda=k(k+1)\log_2 S-k$.
Corollary~\ref{cor:fractalized-set-system} therefore gives the following
explicit tradeoff.

\begin{corollary}
\label{cor:explicit}
For every $1<S\le 2$, TSP admits a bounded-error quantum algorithm using
$S^{n+o(n)}$ QRAM space and $\widehat T(S)^{n+o(n)}$ time,
where
\[
        \widehat T(S)
        =
        \min_{\substack{r\in\Z_{\geq 0}\\ S^{2^r}\le 2}}
        2^{1-2^{-r}}
        \left(
            S^{2^r}\sqrt{\widehat P(S^{2^r})}
        \right)^{2^{-r}} .
\]
The same bound holds for the Hypercube Path query problem and, more generally,
for every bounded-degree min-type permutation problem.
\end{corollary}

\section{Comparison with previous quantum tradeoffs}
\label{sec:comparison}

\begin{figure}
    \centering
    \includegraphics[width=\linewidth]{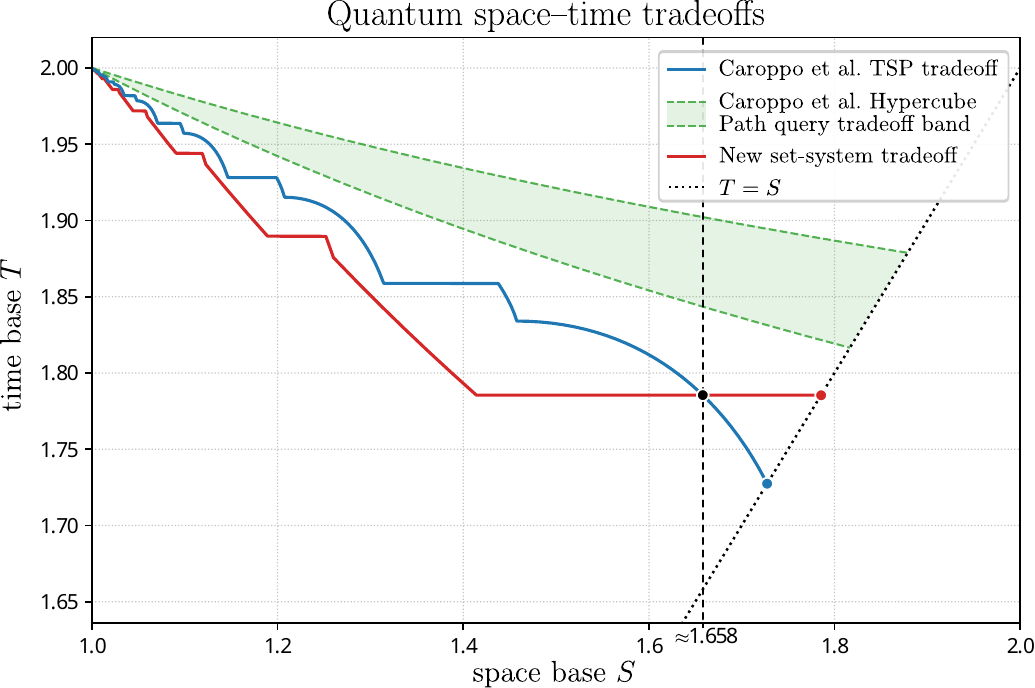}
    \caption{Comparison of quantum space--time tradeoffs for TSP and the
    Hypercube Path query problem.  The shaded band indicates the power-law
    bounds of Caroppo et al.~\cite{CaroppoEtAl2026}.}
    \label{fig:tradeoff_curves}
\end{figure}

Ambainis et al.\ obtained quantum dynamic-programming algorithms for TSP in
$1.728^{n+o(n)}$ time and space, and for general vertex-ordering
problems in $1.817^{n+o(n)}$ time and space
\cite{AmbainisEtAl2019}.  Caroppo et al.\ subsequently studied quantum
space--time tradeoffs for these algorithms, motivated by the cost of QRAM~\cite{CaroppoEtAl2026}.

For divide-and-conquer problems, including TSP, Caroppo et al.\ obtain a
tradeoff with $T$ lying between $2/S^{0.268}$ and $2/S^{0.201}$, depending on the space range and the point from which fractalization is applied.  For the more
general Hypercube Path framework for permutation and vertex-ordering problems,
their best tradeoff curve has $T$ lying between $2/S^{0.161}$ and $2/S^{0.099}$.

The set-system curve used in Figure~\ref{fig:tradeoff_curves} is the explicit
bound $\widehat T(S)$ from Corollary~\ref{cor:explicit}.  Evaluating this
bound gives an improvement over the previous TSP tradeoff for every
$1<S\le 1.657$.  For the Hypercube Path query problem, the explicit curve lies below the
Caroppo et al.\ Hypercube Path band (and thus improves the tradeoff) throughout the nontrivial range $1<S<T$.
The plotted curves correspond to the monotone closure of the stated tradeoffs: if $(S_0,T_0)$ is feasible, then $(S,T_0)$ is feasible for every $S\geq S_0$.

\section{Concluding remarks}

The quantum improvement obtained here comes from replacing enumeration
of the set-system cover by one application of quantum minimum finding. This approach hits a point of diminishing returns around $S\approx 1.43$, beyond which additional QRAM space no longer improves runtime. It would be interesting to explore whether the set-system approach can be combined more tightly with the specialized Ambainis-style quantum dynamic programs for TSP to extend the tradeoff to larger space regimes or improve it more generally. 

For general semiring-valued permutation problems, including counting variants, the outer operation cannot necessarily be implemented directly via quantum minimum finding. Extending the approach would
therefore require a quantum primitive for the semiring addition $\oplus$ or a different use of the set-system cover.

\bibliography{references}

\end{document}